\documentclass[10pt]{scrartcl} 

\usepackage{hyperref}
\usepackage{scrpage2}
\usepackage[utf8]{inputenc}
\usepackage[english]{babel}
\usepackage{amsfonts,amsmath,amssymb,bbm}
\usepackage{amsthm}
\usepackage{geometry} 
\usepackage{dsfont}
\usepackage{enumerate}
\usepackage[numbers]{natbib}
\usepackage{doi}
\usepackage{url}

\theoremstyle{plain}
\newtheorem{thm}{Theorem}[section]
\newtheorem{lem}[thm]{Lemma}

\newtheorem{prop}[thm]{Proposition}
\theoremstyle{definition}

\theoremstyle{remark}
\newtheorem{rem}[thm]{Remark}

\newcommand{\mL}{\mathcal L}
\newcommand{\N}{\mathds{N}}
\newcommand{\Z}{\mathds{Z}}

\newcommand{\R}{\mathds{R}}
\newcommand{\C}{\mathds{C}}

\newcommand{\id}{\mathds{1}}            
\newcommand{\hilbert}{{\mathcal H}}      
\DeclareMathOperator{\re}{Re}            
\DeclareMathOperator{\sign}{sign}        
\DeclareMathOperator{\tr}{tr}            
\newcommand{\e}{\mathrm{e}}   
\newcommand{\I}{\mathrm{i}}    
\newcommand{\di}{\mathrm{d}}    

\long\def\MSC#1\EndMSC{\def\arg{#1}\ifx\arg\empty\relax\else
     {\narrower\noindent%
{2010 Mathematics Subject Classification}: #1\\} \fi}
\long\def\PACS#1\EndPACS{\def\arg{#1}\ifx\arg\empty\relax\else
     {\narrower\noindent%
{PACS numbers}: #1}\fi}
\long\def\KEY#1\EndKEY{\def\arg{#1}\ifx\arg\empty\relax\else
	{\narrower\noindent%
Keywords: #1\\}\fi}

\title{
Anderson's orthogonality catastrophe in one dimension induced by a magnetic field
}
\author{
Hans Konrad Kn\"orr\thanks{hanskonrad.knoerr@fernuni-hagen.de} \and
Peter Otte\thanks{peter.otte@rub.de} \and
Wolfgang Spitzer\thanks{wolfgang.spitzer@fernuni-hagen.de}\\
{\small{FernUniversit\"at in Hagen, Fakult\"at f\"ur Mathematik und Informatik,}}\\[-1ex] {\small{LG Angewandte Stochastik, 58084 Hagen, Germany}}
}
\date{
}
\begin{document}
\maketitle
\vspace{-5ex}
\abstract{According to Anderson's orthogonality catastrophe, the overlap of the $N$-particle ground states of a free Fermi gas with and without an (electric) potential decays in the thermodynamic limit. For the finite one-dimensional system various boundary conditions are employed. Unlike the usual setup the perturbation is introduced by a magnetic (vector) potential. Although such a magnetic field can be gauged away in one spatial dimension there is a significant and interesting effect on the overlap caused by the phases. We study the leading asymptotics of the overlap of the two ground states and the two-term asymptotics of the difference of the ground-state energies. In the case of periodic boundary conditions our main result on the overlap is based upon a well-known asymptotic expansion by Fisher and Hartwig on Toeplitz determinants with a discontinuous symbol. In the case of Dirichlet boundary conditions no such result is known to us and we only provide an upper bound on the overlap, presumably of the right asymptotic order. 
}\\

\KEY
Fermi gas, Toeplitz matrix, Hilbert matrix, Szeg{\H o} limit theorem, Fisher--Hartwig con\-jec\-ture
\EndKEY

\section{Introduction}
In 1967, P.W. Anderson \cite{And-1967(PRL)} discovered that in the limit $N \to \infty$ the ground state of a free gas of $N$ non-interacting fermions is orthogonal to the ground state of the same system but perturbed by an external potential. The asymptotic behaviour of the overlap of the two ground states is of the order $N^{-\gamma}$ with a constant $\gamma>0$ which depends on scattering parameters of the potential. 30 years later, I. Affleck \cite{Aff-1997} related Anderson's orthogonality catastrophe (AOC) to an asymptotic expansion of the difference of the ground-state energies of the free and the perturbed system. Assuming that the latter is of the asymptotic form $c_0 + c_1/N + o(1/N)$ Affleck argued that $c_1=\gamma$. Here, the leading term, $c_0$, is called Fumi term and the coefficient, $c_1$, of the next-to-leading term is called finite-size energy.  

In this paper, we study AOC and the difference of the ground-state energies for a one-dimensional system which is not perturbed by an external potential but rather by an external magnetic field given by a magnetic (vector) potential $a$. We confine the particles to the interval $[-L,L]$, impose periodic or Dirichlet boundary conditions, and consider the thermodynamic limit $N,L \to \infty$ where the particle density $\frac{N}{2L}$ tends to a fixed finite value $\rho>0$. In \cite{LLL-1996}, quantum fluctuations of the current in a metallic loop were studied. This led to a logarithmically divergent term similar to the orthogonality catastrophe in the present paper.

It is well-known that in one spatial dimension the Hamiltonian with magnetic field can be mapped to a free Hamiltonian (with the same or slightly modified eigenvalues as the free problem) by a gauge transformation. Note that while the ground-state energy is changed little (or not at all) by the magnetic field, the effect on the overlap of the ground states is surprisingly strong. Therefore, we cannot expect that Affleck's identification of the AOC exponent with the finite-size energy holds. 

The main part of this paper is dedicated to the asymptotic analysis of the overlap of the ground state of the free Hamiltonian and the one with a magnetic field. For periodic boundary conditions, the leading order of the asymptotic expansion of the overlap of the two ground states is $c N^{-2 \delta^2/\pi^2}$ with a finite constant $c>0$. A precise formulation of this result is given in Theorem~\ref{thm-main}. Here, $\delta$ is half of the integral of $a$ (except for additive multiples of $\pi$). This behaviour is remindful of the one for potentials. In Theorem~\ref{thm: upper bound} we give an asymptotic upper bound of this overlap, namely $\tilde{c} N^{- 2 \sin^2(\delta)/\pi^2}$ with a finite constant $\tilde{c}>0$. This second result holds for periodic as well as Dirichlet boundary conditions. It is noteworthy that the obtained exponent does not depend on the specific form of the magnetic potential $a$ but only on its integral $\delta$ which is related to the magnetic flux. 

For the proof of both theorems we rewrite the inner product of the two ground states as a Toeplitz determinant with an $L$-dependent symbol.
Due to this $L$-dependence, classical Szeg{\H{o}} limit theorems including those for Fisher--Hartwig symbols cannot be applied directly.
The crucial step of our proofs is that this determinant is factored into a determinant, which is asymptotically of order 1, times a Toeplitz determinant with a discontinuous $L$-dependent symbol.
This discontinuous symbol in turn can be interpreted as corresponding to an effective flux which looks like the one induced by Dirac's $\delta$-function as a magnetic potential.
In the case of periodic boundary conditions the leading $N$-asymptotics of this particular Toeplitz determinant with discontinuous symbol is given by a Toeplitz determinant with an $L$-independent symbol of the Fisher--Hartwig class.
This exact Toeplitz determinant has been analysed by Fisher and Hartwig in their 1968 paper, which they used to bring forward a more general conjecture on the asymptotics for a certain kind of discontinuous and singular symbols \cite{FH-1968}. Using their result we find the leading asymptotics of the overlap of the respective ground states. An upper bound is determined by using the inequality $\det(\id+A) \leq \e^{\,\tr\! A}$, where the exponent on the right-hand side is the so-called Anderson integral. We study this Anderson integral to obtain the asymptotic upper bound for periodic boundary conditions. For Dirichlet boundary conditions we choose a different approach which requires the asymptotic analysis of a determinant of a Toeplitz plus a Hilbert matrix.

In the following section we specify the quantum mechanical model and fix the notation. As a first simple result we obtain the Fumi term and the finite-size energy in Subsection~\ref{sec: FSE}. The precise formulation of our main results on Anderson's orthogonality catastrophe, Theorems~\ref{thm-main} and \ref{thm: upper bound}, is given in Subsection~\ref{sec: AOC}. In Section~\ref{sec: Toeplitz} we prove these two theorems. For this purpose we introduce the notion of generalised Toeplitz matrices. In Lemma~\ref{lem-main} we give a decomposition of the overlap of the ground states into a bounded part and a Toeplitz determinant with discontinuous symbol as mentioned above. This lemma is proved in Section~\ref{sec: decomposition}. In Subsection~\ref{sec: periodic proof} we show Theorem~\ref{thm-main}. Subsections~\ref{sec: proof 2} is devoted to the proof of Theorem~\ref{thm: upper bound}.

Finally, we comment on the mathematical work in this field all of which is rather recent. The papers \cite{KOS-2013,GKM-2014,GKMO-2014A} focus on AOC (caused by potentials and not magnetic fields) and prove upper bounds on the overlap which are all in agreement with Anderson's bounds. In \cite{Geb-2014A} and \cite{OS-2014WiP}, the finite-size energy and its relation to the coefficient $\gamma$ in AOC is studied in a one-dimensional Fermi system. Yet, we do not know whether Anderson's prediction about the decay rate $\gamma$ of the overlap is sharp. However, in the model of this paper we are able to prove the sharp decay rate. But we also notice that Affleck's correspondence between the coefficient $\gamma$ and the finite-size correction does not hold here. This should not diminish the value of this fascinating relation in the case of a potential.

\section{Model and results}

\subsection{The Hamiltonian for a single particle}\label{sec: 1fermionenergy}

We consider a single, spinless particle confined to the interval $\left[-L,L\right] \subset \R$.
The Hilbert space is $\hilbert_L := \mL^2([-L,L])$, equipped with the usual inner product given by
$\langle f,g \rangle := \int_{-L}^L \overline{f}(x)\, g(x)\,\di x$. 
We study the system for periodic and for Dirichlet boundary conditions
where arguments and results for the latter can be taken over to Neumann boundary conditions
(or more general Robin boundary conditions).
We denote the eigenfunctions, eigenvalues etc.~for different boundary conditions by the same respective symbols.
It will become clear from the context which boundary conditions are used.

\paragraph{Periodic boundary conditions}
The Hamiltonian of a single free particle on this interval is
$H_{L} := -\frac{\di^2}{\di x^2}$ with periodic boundary conditions. 
The energy eigenvalues denoted by $\lambda_L$ and the
corresponding eigenstates are determined by
\begin{align}
&-\frac{\di^2 \varphi_L}{\di x^2}(x) = \lambda_L\varphi_L(x) \,,\ x \in \left]-L,L\right[\,,\label{eq: free de}\\
&\varphi_L(-L) - \varphi_L(L)=0\,,\quad\frac{\di \varphi_L}{\di x} (-L) - \frac{\di \varphi_L}{\di x} (L)=0\,.\label{eq: free periodic bc}
\end{align}
Then we have, normalised to one in $\mL^2$-norm,
$\varphi_{L,j} (x) = (2L)^{-\frac12} \exp\left(- \I \frac{\pi j}{L}x\right)$ and $\lambda_{L,j} = \left( \frac{\pi j}{L} \right)^2$, $j \in \Z$.
Note that all eigenvalues except for $\lambda_{L,0}=0$ are twofold degenerate and
the eigenvalue 0 is non-degenerate.

Let $a \in W^{1,1}(\R)$, the Sobolev space of integrable and continuously differentiable
functions with integrable derivatives.
We denote the multiplication operator given by $(af)(x)=a(x)f(x)$ by $a$ as well.
In order to have a well-defined operator on $\hilbert_L$ satisfying periodic boundary conditions,
we additionally assume
$a(L) = a(-L)$.
The Hamiltonian of the system with the external magnetic (vector) potential $a$ is
  $H_{a,L} := (-\I \frac{\di}{\di x} - a)^2$.
The spectrum $\sigma(H_{a,L})$ is determined by the eigenvalue problem
\begin{align}
&\left(-\I \frac{\di}{\di x} - a(x)\right)^2\psi_L(x) = \mu_L \psi_L(x)\,,\ x \in \left]-L,L\right[\,,\label{eq: perturbed de}\\
&\psi_L(-L)-\psi_L(L)=0\,,\quad
\frac{\di \psi_L}{\di x} (-L) - \frac{\di \psi_L}{\di x}(L)=0\,.\label{eq: perturbed periodic bc}
\end{align}
This eigenvalue problem can be solved using the gauge transformation $\e^{\I \Phi_L}$
where $\Phi_L$ is given by
\begin{align}\label{eq-magneticflux}
\Phi_L(x) := \frac12 \int\limits_{-L}^x a(y)\,\di y - \frac12 \int\limits_x^L a(y)\,\di y \,,\ x \in \left[-L,L\right]\,.
\end{align}
The quantity $\Phi_L(x)$ is called magnetic flux and represents the phase shift due to the magnetic field.
Since $\Phi_L'(x) = a(x)$ and $a(L)=a(-L)$,
the corresponding eigenvalues and eigenfunctions are easily determined:
$\mu_{L,j} = \left( \frac{j \pi + \Phi_L(L)}{L}\right)^2$ and
$\psi_{L,j}(x) = (2L)^{-\frac12} \exp\!\left( \I \Phi_L(x)-\I \frac{j\pi+\Phi_L(L)}{L} x\right)$, $j \in \Z$.
On the one hand, all eigenvalues except for 0, which is non-degenerate, are degenerate with multiplicity two
if $\Phi_L$ is an integral multiple of $\pi$ and, if $\Phi_L(L)$ is a half-integral multiple of $\pi$, the eigenvalue 0 is also
degenerate with multiplicity two.
On the other hand, in any other case all eigenvalues are non-degenerate.

\paragraph{Dirichlet boundary conditions}
For Dirichlet boundary conditions, Eqs.~\eqref{eq: free periodic bc} and \eqref{eq: perturbed periodic bc}
have to be replaced by 
\begin{align*}
 \varphi_L(-L) =0 = \varphi_L(L)\quad \text{and}\quad \psi_L(-L)=0 =\psi_L(L)\,,
\end{align*}
respectively.
Moreover, we only assume $a \in W^{1,1}(\R)$, but not that $a(L)=a(-L)$.
The free and the magnetically perturbed Hamiltonian have the same spectrum,
$\lambda_{L,j} = \mu_{L,j} = \left( \frac{\pi j}{2L} \right)^2$
with $j \in \N$.
This is easily shown using the gauge transformation $\e^{\I \Phi_L}$
where $\Phi_L$ is defined as above.
The corresponding normalised eigenfunctions are given by
\begin{align*}
\varphi_{L,j} (x) := \begin{cases} \frac{1}{\sqrt{L}} \sin\!\left(\frac{j \pi}{2L}x\right)\ &\text{for even}\ j\\
   \frac{1}{\sqrt{L}} \cos\!\left(\frac{j \pi}{2L}x\right)\ &\text{for odd}\ j\end{cases}\quad \text{and}\quad
\psi_{L,j} := \e^{\I \Phi_L} \varphi_{L,j}\,.
\end{align*} 
In contrast with the periodic case, all eigenvalues are non-degenerate and $0$ is not an eigenvalue.\\

\begin{rem}
One could regard $b := a'$ as the magnetic field
(whatever the physical interpretation of a magnetic field in one dimension is).
Due to the condition that $a$ is integrable, $b$ has to decay to zero at infinity
which, for instance, rules out the case of a constant magnetic field.
Moreover, the integral of $b$ has to vanish. In the periodic case the additional restrictions $b(-L) = -b(L)$ and $\int_{-L}^Lb(x)\,\di x = 0$ for any $L$ occur.
\end{rem}

We want to remind the reader that the gauge transformation maps
the Hamiltonian with magnetic field, Eq.~\eqref{eq: perturbed de}, to a free Hamiltonian
(with the same or slightly modified eigenvalues as the free problem \eqref{eq: free de}).
However, the eigenfunctions depend on this gauge.
Therefore the effect of the magnetic potential on the overlap of the ground states can be
expected to be more drastic than the one on the difference of the ground-state energies.

\subsection{$N$ non-interacting fermions and their ground-state energy}\label{sec: FSE}

\paragraph{Periodic boundary conditions}
Now we consider a system of $N$ non-interacting spinless fermions
on the interval $\left[-L,L\right] \subset \R$ with periodic boundary conditions.
The corresponding $N$-particle Fock space is the $N$-fold antisymmetric tensor product
of the one-particle Hilbert space, $\mathcal{F}_L^{(N)} := {\wedge}^N \hilbert_L$.
The energy of this system is represented by the Hamiltonian
$H_{a,N,L} := {\bigotimes}^N H_{a,L}$.
Note that for $a \equiv 0$ we recover the free Hamiltonian.
The ground-state energy $E_{a,N,L}$ as well as corresponding eigenfunctions of these
two $N$ fermion system can be calculated from the eigenvalues and -functions of the one-particle system
as derived in Subsection~\ref{sec: 1fermionenergy}.
The magnetic flux at position $L$ is decomposed as
\begin{align}\label{eq-2}
\Phi_L(L) = n_L\pi + \delta_L
\end{align}
with $n_L \in \Z$ and $\delta_L \in \left]-\frac{\pi}{2},\frac{\pi}{2}\right]$.
With this notation the ground-state energy of $H_{a,N,L}$ is
$E_{a,N,L} = \sum_{j\in\mathcal{N}_{n_L}} \left(\frac{\pi j - \Phi_L(L)}{L}\right)^2$\,.
The set $\mathcal{N}_{n_L}$ is $\left\{-m-n_L,...,m-n_L\right\} \subset \Z$
for odd $N$ and
$\left\{-m-n_L,\dots,m-n_L-1\right\} \subset \Z$ for even $N$.
Here and henceforth we set $m:=\frac{N}{2}$ for even $N$ and $m:=\frac{N-1}{2}$ for odd $N$.
Note that the ground-state energy is twofold degenerate in the cases
$\delta_L = \frac{\pi}{2}, N \in (2\N-1)$ and $\delta_L=0, N \in (2\N)$ and
non-degenerate otherwise. Then the difference $E_{a,N,L}-E_{0,N,L}$ of the ground-state energies is equal to
$\frac{\delta_L^2 N}{L^2}$ for odd $N$ and to $\frac{\delta_L (\delta_L-\pi) N}{L^2}$ for even $N$.

Therefore the sequence $\left\{E_{a,N,L}-E_{0,N,L}\right\}_{N,L}$ has two limit points and
the asymptotic behaviour differs for the subsequences with even and odd values of $N$, respectively.
We find
\begin{align*}
E_{a,N,L} - E_{0,N,L} =
\begin{cases} \frac{4 \delta^2 \rho^2}{N} + o\big(\frac{1}{N}\big) & \text{ for odd } N\,\\
\frac{4 \delta (\delta-\pi) \rho^2}{N} + o\big(\frac{1}{N}\big) & \text{ for even } N
\end{cases}
\end{align*}
as $N,L\to\infty$ with $\frac{N}{2L}\to \rho$ for a fixed limit density $\rho > 0$.
Here the parameter $\delta \in \left]-\frac{\pi}{2},\frac{\pi}{2}\right]$ is given by 
$n \pi + \delta = \frac12 \int_{-\infty}^{\infty} a(y) \, \di y$ where $n \in \Z$ is chosen appropriately.
Thus $\delta = \lim_{L \to \infty} \delta_L$.

Note that the leading term of the asymptotic expansion, which is of order 1 and called Fumi term, is zero.
However, the next term in this expansion which is of order $1/N$, the finite-size energy, is, in general,
non-zero for finite values of $N$.

\begin{rem}
We want to emphasize that the asymptotic terms may depend, in general, how we perform the thermodynamic limit,
cf. \cite{Geb-2014A,OS-2014WiP} for the case with an (electric) potential.
\end{rem}

Moreover, we find that the ground state of the Hamiltonians $H_{a,N,L}$ is the Slater determinant
$\Psi_{N,L} := \bigwedge_{k \in \mathcal{N}_{n_L}} \psi_{L,k}\,.$
Note that $\psi_{L,k-n_L}(x) = \e^{\I \left(\Phi_L(x) - \frac{\delta_L}{L} x\right)}\varphi_{L,k}(x)$.

\paragraph{Dirichlet boundary conditions}
Analogously to the periodic case we define the $N$-particle Hamiltonians
$H_{0,N,L}$ and $H_{a,N,L}$ with Dirichlet boundary conditions.
The corresponding ground-state energies are
$E_{0,N,L} = E_{a,N,L} = \sum\limits_{j=1}^N \left( \frac{j \pi}{2L}\right)^2$\,.
They are independent of $a$ and therefore the difference of the ground-state energies vanishes for all particle numbers $N$.
The ground state of the perturbed system is the Slater determinant
$\Psi_{N,L} = \bigwedge_{j=1}^N \psi_{L,j} = \bigwedge_{j=1}^N \left(\e^{\I \Phi_L} \varphi_{L,j}\right)$\,.

\subsection{Anderson's orthogonality catastrophe}\label{sec: AOC}

We mainly study the asymptotic behaviour of the overlap of the ground states of the free and the perturbed system.
In the case of periodic boundary conditions we can determine the asymptotics of the overlap rather explicitly:
\begin{thm}[Anderson's orthogonality catastrophe, periodic boundary conditions]\label{thm-main}
Let $\rho > 0$ and $N \in \N$.
Assume $\int_{-\infty}^{\infty} |ya(y)|\, \di y < \infty$, $\max\left\{|n_L|\right\}<\infty$,
$|\delta_L| < \frac{\pi}{2}$, at least for a subsequence of $L$'s, with $|\delta|<\frac{\pi}{2}$.
Then there exists a constant $c > 0$ such that the overlap of the two $N$-fermion ground states satisfies
\begin{align}\label{eq: AOC}
\left|\det \!\left( \langle \varphi_{L,j}, \psi_{L,k} \rangle \right)_{j\in\mathcal{N}_0,k\in\mathcal{N}_{n_L}}
\right|^2 = N^{-\frac{2 \delta^2}{\pi^2}} (c + o(1)) 
\end{align}
in the thermodynamic limit $N,L \to \infty$ with $\frac{N}{2L} \to \rho$. In particular, the overlap vanishes in
the considered limit for $\delta \neq 0$. For $\delta = 0$ the overlap is equal to 1.
\end{thm}
The proof of this theorem is the main task of this paper and is carried out in Section~\ref{sec: Toeplitz}.
For our proof it is crucial that the matrix on the left-hand side of Eq.~\eqref{eq: AOC} is a Toeplitz matrix.
The leading term of the asymptotics results from a Szeg{\H o}--type limit theorem for a Toeplitz determinant
with a discontinuous symbol as given by Fisher and Hartwig \cite{FH-1968}.
\begin{rem}
Levitov, Lee and Lesovik derived a similar result to Eq.~\eqref{eq: AOC} in a time-dependent setup \cite{LLL-1996, LL-1993}. They consider a loop, i.\,e. an interval of length $L$ with periodic boundary conditions, in a fixed external potential. By a short pulse a magnetic flux is induced which leads to a time-dependent vector potential in the Hamiltonian. The overlap of the ground states of the system long before the pulse and long afterwards are compared where the thermodynamic limit $L \to \infty$ is taken at fixed Fermi energy.
\end{rem}

For reasons that will become clear soon, we, moreover, show a certain upper bound to the overlap
of the ground states for both boundary conditions. The bound is obtained with the inequality $\ln\!\left(\det A\right) \leq - \tr(\mathds{1}-A)$ and calculating the leading term of the so-called Anderson integral. For Dirichlet boundary conditions the direct calculation of the Anderson integral is omitted in favour of another approach using a property of a Hilbert matrix appearing in the asymptotic analysis of the determinant.

\begin{thm}[Upper bound for Anderson's orthogonality catastrophe, periodic and Dirichlet boundary conditions]\label{thm: upper bound}
Let $\rho > 0$ and $N \in \N$. With the same assumptions as in Theorem~\ref{thm-main},
an upper bound for the asymptotics of the overlap is given by
\begin{align}\label{eq-15}
\left|\det \!\left( \langle \varphi_{L,j}, \psi_{L,k} \rangle \right)_{j\in\mathcal{N}_0,k\in\mathcal{N}_{n_L}} \right|^2
\leq N^{-\frac{2}{\pi^2} \sin^2(\delta)} (\widetilde{c} + o(1)) 
\end{align}
with a constant $\widetilde{c} > 0$ as $N,L \to \infty$ with $\frac{N}{2L} \to \rho$.
\end{thm}
The proof of this theorem is given in Subsection~\ref{sec: proof 2}.
\begin{rem}
Up to the constant $\frac{2}{\pi^2}$, the negative exponent $\delta^2$ of the exact asymptotics in Eq.~\eqref{eq: AOC}
and the negative exponent $\sin^2(\delta)$ of the upper bound in Ineq.~\eqref{eq-15} satisfy ``Anderson's rule''. 
So they are essentially the sine respectively the arc sine of each other.
 Note that the negative exponent $\delta^2$ is sharp and cannot be improved.
\end{rem}
\allowdisplaybreaks
\begin{rem}\label{rem: Dirichlet}
When considering the overlap for Dirichlet boundary conditions, one is, in leading order, led to the asymptotic analysis of
\begin{align*}
\det\!\left( \mathds{1} - \frac{\sin^2(\delta_L)}{\pi^2} P_N H_{\eta}^2 P_N \right)\ \text{as}\ N,L \to \infty\,.
\end{align*}
Here $H_{\eta} := \left( \frac{1}{j+k+\eta} \right)_{j,k \in \N}$ is a Hilbert matrix with $\eta \in \R$, $-\eta \notin \N$, and $P_N$ a projection on an $N$-dimensional subspace.
In order to determine this asymptotics, we would need a Szeg{\H o} limit theorem for the determinant of the sum of the Toeplitz matrix $\mathds{1}$ and the square of the Hilbert matrix $H_{\eta}$.
We are only aware of asymptotic results for the sum of a Toeplitz and a Hankel matrix requiring a certain relation between the two symbols which, however, do not apply here.
It seems reasonable that the overlap is independent of boundary conditions and in particular the same as for periodic boundary conditions. 
\end{rem}

\begin{rem} If we would allow a magnetic field $b$ with non-vanishing (finite) integral then from the start we can only allow Dirichlet boundary conditions. The logarithm of the overlap is now to leading order of the form $-N \left|\int b\right|\!/(2\pi)$. It seems that the next-to-leading term is of the order $\ln(N)$ and that it is of the same form (possibly up to a constant) as above if we replace $b$ by $b-\int b$ in the above asymptotics for vanishing total magnetic field $b$. 
\end{rem}
\section{Asymptotic analysis of Toeplitz determinants}\label{sec: Toeplitz}

Let $f:[-L,L]\to\C$ be a complex-valued (essentially bounded) function.
The operator of multiplication by $f$ is denoted by $f$ as well.
A generalised Toeplitz matrix with symbol $f$ is defined as the $N\times N$-matrix
\begin{align*}
  T_N(f) := \left( \langle\varphi_j,f\varphi_k\rangle \right)_{j,k=1,\ldots, N}
\end{align*}
where $\left\{\varphi_k\right\}_{k}$ is an orthonormal basis of the underlying Hilbert space.
The corresponding determinant $\det(T_N(f))$ is referred to as the (generalised) Toeplitz determinant.
Here the generalisation is with respect to the choice of the basis.
In our context we choose the eigenbasis of the corresponding free Hamiltonian.
In the case of periodic boundary conditions this basis consists of plane waves
and the matrix elements are the Fourier coefficients of $f$.
In particular, they only depend on the difference $j-k$ of the indices and we recover a Toeplitz matrix in the usual sense.
For Dirichlet boundary conditions, the emerging matrix is a sum of a usual Toeplitz and a Hankel matrix.
For a deeper study and an overview of results for Toeplitz and related matrices see, for instance, \cite{BS-1990,DIK-2013}.

Hereinafter, we study symbols of the form $f = \e^{\I g_L}$.
In the case of periodic boundary conditions $g_L$ is given by
\begin{align}
g_L(x):= \Phi_L(x) - \frac{\delta_L}{L}x\quad \text{for}\quad x \in \left[-L,L\right]
\end{align}
with $\Phi_L$ and $\delta_L$ as defined in Eqs.~\eqref{eq-magneticflux} and \eqref{eq-2}, respectively.
The function $g_L$ has a jump of height $2 n_L \pi$ at the boundary $x=\pm L$.
This becomes particularly simple if we take, at least informally, the discontinuous magnetic potential
$\tilde{a}(x):=2 \Phi_L(L) \delta(x)$.
More precisely, we set
\begin{align}\label{eq-4}
  \tilde{g}_L(x) := \Phi_L(L) \sign(x) - \frac{\delta_L}{L}x\quad \text{for}\quad x \in \left[-L,L\right]
\end{align}
which is reminiscent of an integral of $\tilde{a}$.
As $g_L$, $\tilde{g}_L$ has a jump discontinuity of height $2n_L\pi$
at the boundary $x=\pm L$, but also at $x=0$ with height $2\Phi_L(L)$.\\

For Dirichlet boundary conditions we do not have the term $- \frac{\delta_L}{L}x$\,.
Hence we set
\begin{align}
g_L(x):=\Phi_L(x)\quad \text{and}\quad \tilde{g}_L(x) := \Phi_L(L) \sign(x)\quad \text{for}\quad x \in \left[-L,L\right]\,.
\end{align}
Here the jump discontinuities at 0 and at the boundary are of the same height $2 \Phi_L(L)$\,.\\

Now we have
$\left| \det\!\left(\langle \varphi_{L,j}, \psi_{L,k} \rangle \right)_{j,k=1,\dots,N} \right|^2 = \left| \det\!\left(T_N\big(\e^{\I g_L}\big)\right) \right|^2$\,.
Note that the symbol of this Toeplitz determinant depends on the parameter $L$, which also tends to infinity, unlike in
classical Szeg{\H{o}} limit theorems.
To circumvent this problem $D_{N,L} := \det\!\left(T_N\big(\e^{\I g_L}\big)\right)$ is factorised
into the Toeplitz determinant $\widetilde{D}_{N,L} := \det\big(T_N\big(\e^{\I \tilde{g}_L}\big)\big)$
with the symbol $\e^{\I \tilde{g}_L}$ and a bounded term as the main step in the proof of the theorems.
This Toeplitz determinant $\widetilde{D}_{N,L}$ in turn can be rewritten as a usual Toeplitz determinant
with a Fisher--Hartwig symbol times a remainder term of order 1.
More precisely, we use the following lemma.

\begin{lem}\label{lem-main}
For the system with periodic or Dirichlet boundary conditions assume $\int_{-\infty}^{\infty} |y\,a(y)|\,\di y < \infty$, $\max\{|n_L|\} < \infty$,
and $|\delta_L| \leq c$ for some $0 \leq c < \frac{\pi}{2}$.
Then we have
\begin{align*}
\left|  \det\!\left(T_N\big(\e^{\I g_L}\big)\right) \right|^2 = {C}_{N,L} \cdot \left| \det\!\left(T_N\big(\e^{\I \tilde{g}_L}\big)\right) \right|^2
\end{align*}
for any $N$ and $L$. In general, $C_{N,L}$ depends on $N$ and $L$.
Moreover, there are positive constants $\rho_0,\ \delta_1$, and $\delta_2$ such that
\begin{align*}
0 < \delta_1 \leq C_{N,L} \leq \delta_2 < \infty \quad \text{for all}\ N\ \text{and}\ L\ \text{with}\quad \frac{N}{2L} \leq \rho_0\,. 
\end{align*}
\end{lem}

The proof of this lemma is postponed to Section~\ref{sec: decomposition}. It is based on a factorisation of the determinant
and a decomposition of some of the resulting matrices.

\begin{rem}
For simplicity we henceforth assume that the quotient $\frac{N}{2L}$ is kept at a fixed finite value $\rho > 0$ when we take the
thermodynamic limit.
However, this condition may be relaxed to $\lim_{N,L \to \infty} \frac{N}{2L}=\rho$ and the following statements still hold true.
\end{rem}

\subsection{Proof of Theorem~\ref{thm-main}}\label{sec: periodic proof}

In the following we show Theorem~\ref{thm-main} using the decomposition of the determinant as given in Lemma~\ref{lem-main}.

\begin{proof}[Proof of Theorem~\ref{thm-main}]
According to Lemma~\ref{lem-main} the leading asymptotics results from the Toeplitz determinant with the discontinuous symbol
$\e^{\I \tilde{g}_L}$ with $\tilde{g}_L$ given by Eq.~\eqref{eq-4}.
Thus the leading term arises from
\begin{align}
\left| \widetilde{D}_{N,L} \right|^2 &= \left| \det\!\left(\langle \varphi_{L,j}, \e^{\I \tilde{g}_L} \varphi_{L,k} \rangle \right)_{j\in\mathcal{N}_0,k\in\mathcal{N}_{n_L}} \right|^2\notag\\
&= \left| \det\!\left(\frac{\sin(\delta_L)}{\delta_L - \pi (j-k)}\right)_{j,k=-m+\chi_{2\N}(N),\dots,m} \right|^2\,\label{eq-ToeplitzEintraege}
\end{align}
where $m = \lfloor \frac{N}{2}\rfloor$\,.
This can be verified by a straightforward computation using Eq.~\eqref{eq-2}.
With the abbreviation $s_m:= \frac{\sin(\delta_L)}{\delta_L-\pi m}$ we obtain the $N\times N$-matrix
\begin{align*}
  T_N(\e^{\I \tilde{g}_L}) =
\begin{pmatrix}
 s_{0}  & s_{-1} & s_{-2} & \cdots & s_{-N+2} & s_{-N+1}\\
 s_{1}  & s_{0}  & s_{-1} & \cdots & s_{-N+3} & s_{-N+2}\\
\vdots  &           &             &            &                 & \vdots    \\
s_{N-1}& s_{N-2}& s_{N-3} & \cdots & s_{1}  & s_0
\end{pmatrix}\,.
\end{align*}
This is a Toeplitz matrix with the symbol $\e^{\I \tilde{g}_L}$ which is discontinuous at zero and at the boundary, but has no singularities.

If $a$ is compactly supported, $\delta_L$ is constant for sufficiently large values of $L$ and the constant value is $\delta$.
If we relax the condition that $a$ has compact support, the $L$-dependence of the symbol has to be taken into account.
However, with $\|Xa\|_1 < \infty$ and $\|a\|_1 < \infty$, $a$ has to decay faster than $\frac{1}{x^2}$ at infinity.
Therefore $\left| \delta_L - \delta\right| = o(\frac{1}{L})$ as $L \to \infty$.
Thus the first error term decays faster than $L^{-1}$.
Consequently, we may assume the magnetic potential to be compactly supported from now on.
In this special case there is an $\widetilde{L} < \infty$ such that $\delta_L$
is constant (i.e. does not depend on $L$) and is equal to $\delta$ for every $L > \widetilde{L}$. 

By the above arguments, the parameter $\delta_L$ can be considered to have the constant value $\delta=\lim_{L \to \infty} \delta_L$
for sufficiently large $L$.
The Toeplitz determinant we obtain this way has already been considered by Fisher and Hartwig in 1968 \cite{FH-1968}
as an example where the conjecture they brought forward in the same paper, now called Fisher--Hartwig conjecture, holds.
They explicitly calculate the leading asymptotic term in the limit $N \to \infty$ using Cauchy's determinant formula.
So we find
\begin{align*}
\left|\det\!\left( T_N(\e^{\I \tilde{g}_{L}})\right)\right|^2 \sim E N^{-2\frac{\delta^2}{\pi^2}} (1+{o}(1))
\end{align*} 
as $N,L \to \infty$ with $\frac{N}{L} \to \rho$.
Note that the constant $E$ depends on $\delta$ and is given explicitly in \cite{FH-1968}.
Related results on the Fisher--Hartwig conjecture can be found in \cite{Bas-1979, Wid-1973, BS-1985}.
Merging the above results and Lemma~\ref{lem-main} we obtain the assertion of Theorem~\ref{thm-main}.
\end{proof}

\begin{rem}
In his study on Anderson's orthogonality catastrophe \cite{Llo-1971}, P.~Lloyd constructs a model unitary transformation
between the eigenfunctions of the free Hamiltonian $H_0$ and the perturbed Hamiltonian $H_0+V$.
The matrix elements \cite[Eq.~(A.2.4)]{Llo-1971}
of this unitary transformation with respect to the free eigenfunctions
are of the same form as in \eqref{eq-ToeplitzEintraege} with $\delta_L$ replaced by the phase factor of the potential,
which allowed him to justify the exponent in Anderson's orthogonality catastrophe.
\end{rem}

\subsection{Proof of Theorem~\ref{thm: upper bound}}\label{sec: proof 2}

\paragraph{Periodic boundary conditions and the Anderson integral}

We first show Theorem~\ref{thm: upper bound} for periodic boundary conditions.
To this end we study the asymptotic behaviour of the corresponding Anderson integral.
Again, it suffices to study the asymptotics of the Toeplitz determinant, $\widetilde{D}_{N,L}$,
with the discontinuous symbol $\e^{\I \tilde{g}_L}$. 
For simplicity we assume $n_L=0$ and $N$ to be odd in the following and set $m:=\frac{N-1}{2}$.
For $n_L \neq 0$ the area of summation has to be shifted by $n_L$ to the left and
for even $N=2m$ the term with $j=m$ has to be removed from the sum.
With the two projections 
$
P_N := \sum\limits_{j=-m}^m \langle \varphi_{L,j}, \,\cdot\,\rangle \varphi_{L,j} \quad \text{and} \quad
\widetilde{\Pi}_N := \sum\limits_{j=-m}^m \langle \e^{\I \tilde{g}_L} \varphi_{L,j}, \,\cdot\,\rangle \e^{\I \tilde{g}_L} \varphi_{L,j}
$
we have
\begin{align*}
\left|\widetilde{D}_{N,L}\right|^2 = \left| \det\!\left( \langle \varphi_{L,j},\e^{\I \tilde{g}_L} \varphi_{L,k} \rangle \right)_{j,k=-m,\dots,m} \right|^2
= \det\!\left(P_N \widetilde{\Pi}_N P_N\right)\,.
\end{align*}
The inequality $\det(A) \leq \exp\!\left(- \tr(\mathds{1}-A) \right)$, which holds for trace class operators $A$, yields
\begin{align*}
\det\!\left(P_N \widetilde{\Pi}_N P_N\right) \leq \exp\!\left(- \tr\!\left(P_N(\mathds{1} - \widetilde{\Pi}_N)P_N\right) \right)
                                                   = \exp\!\left( - \mathcal{I}_{N,L} \right)
\end{align*}
where
$
\mathcal{I}_{N,L} := \sum\limits_{j=-m}^m \left[ \sum\limits_{k=m+1}^{\infty} \left| \langle \varphi_{L,j},
\e^{\I \tilde{g}_L} \varphi_{L,k} \rangle \right|^2 + \sum\limits_{k=-\infty}^{-m-1} \left| \langle \varphi_{L,j},
\e^{\I \tilde{g}_L} \varphi_{L,k} \rangle \right|^2 \right]
$
is the Anderson integral.
Since
$\langle \varphi_{L,j},\e^{\I \tilde{g}_L} \varphi_{L,k} \rangle = \frac{\sin(\delta_L)}{\pi(j-k)+\delta_L}
$, 
the individual summands of the Anderson integral only depend on the difference of the indices $j-k$.
Rescaling and renaming the summation variables we obtain
\begin{align*}
\frac{\pi^2}{\sin^2(\delta_L)} \mathcal{I}_{N,L} &= \sum\limits_{j=1}^N
\sum\limits_{k=1}^{\infty} \frac{1}{\left(k+j-1-\frac{\delta_L}{\pi}\right)^2}
+ \sum\limits_{j=1}^N \sum\limits_{k=1}^{\infty} \frac{1}{\left(-k-j+1-\frac{\delta_L}{\pi}\right)^2}\\
                             &=\sum\limits_{k=1}^N \frac{k}{\left(k-1-\frac{\delta_L}{\pi}\right)^2}
+ N \sum\limits_{k=N+1}^{\infty} \frac{1}{\left(k-1-\frac{\delta_L}{\pi}\right)^2} \notag\\
&\qquad+ \sum\limits_{k=1}^N \frac{k}{\left(-k+1-\frac{\delta_L}{\pi}\right)^2}
+ N \sum\limits_{k=N+1}^{\infty} \frac{1}{\left(-k+1-\frac{\delta_L}{\pi}\right)^2}\,.
\end{align*}
On the one hand, the second and the fourth sum on the right-hand side are mono\-to\-nous\-ly
decaying at least as $\frac{1}{N}$ for large $N$.
On the other hand, the first and the third sum are diverging for $N \to \infty$
and their respective asymptotics are $\ln(N) + O(1)$.
Hence the leading term of the asymptotics of the Anderson integral is
$
\mathcal{I}_{N,L} = \frac{2}{\pi^2} \sin^2(\delta) \ln(N) + O(1)
$
as $N,L \to \infty$ with $\frac{N}{2L} \to \rho>0$.
Together with Lemma~\ref{lem-main} the above derivation yields
\begin{align*}
\left|\det \big( \langle \varphi_{L,j}, \psi_{L,k} \rangle \big)_{j,k=-m,\ldots,m} \right|^2
\leq N^{-\frac{2}{\pi^2} \sin^2(\delta)} (\widetilde{c} + o(1))\,,
\end{align*}
thus finishing the proof of Theorem~\ref{thm: upper bound} for periodic boundary conditions.

\paragraph{Dirichlet boundary conditions and the Hilbert matrix}

For the Dirichlet case we use a different way to prove the respective assertion of Theorem~\ref{thm: upper bound}.
This approach has the advantage that we have equalities until the very last step where the only estimate of this proof comes into play.
It shows a possible way to tackle the problem of the exact asymptotics of the overlap
as already mentioned in Remark~\ref{rem: Dirichlet}.

For simplicity we consider the particle number to be even and set $N=:2M$ throughout this subsection.
For odd particle numbers $N$ some of the following formulae are slightly more complicated, but the general method works as well.
In order to further simplify $\widetilde{D}_{N,L} := \det T_N(\e^{\I \tilde{g}_L})$ we have a closer look at the matrix elements
$\langle\varphi_{j},e^{i \tilde{g}_{L}}\varphi_{k} \rangle$
that appear in this determinant.
By the trigonometric addition theorems we find
\begin{align*}
  \varphi_{L,j}(x)\varphi_{L,k}(x)
   & =
\begin{cases}
  \frac{1}{2L} \left( \cos\!\left(\frac{\pi (j-k)}{2L}x\right) - \cos\!\left(\frac{\pi (j+k)}{2L}x\right) \right) & \text{for}\ j,k \ \text{even}\,, \\
  \frac{1}{2L} \left( \cos\!\left(\frac{\pi (j-k)}{2L}x\right) + \cos\!\left(\frac{\pi (j+k)}{2L}x\right) \right) & \text{for}\ j,k \ \text{odd}\,, \\
  \frac{1}{2L} \left( \sin\!\left(\frac{\pi (j-k)}{2L}x\right) + \sin\!\left(\frac{\pi (j+k)}{2L}x\right) \right) & \text{for}\ j \ \text{even and}\ k\ \text{odd}\,, \\
  \frac{1}{2L} \left( - \sin\!\left(\frac{\pi (j-k)}{2L}x\right) + \sin\!\left(\frac{\pi(j+k)}{2L}x\right) \right) & \text{for}\ j\ \text{odd and}\ k \ \text{even}\,.
\end{cases}
\end{align*}
We obtain
$\int\limits_{-L}^L e^{\I \tilde{g}_L(x)} \cos\!\left(\frac{\pi m}{2L}x\right)\di x = 0$ 
\ \ for\ \ $m \neq 0$, 
$\int\limits_{-L}^L e^{\I \tilde{g}_L(x)} \cos\!\left(\frac{\pi m}{2L}x\right) \di x = 2 L \cos\!\left(\Phi_L(L)\right)$\ \
for\ \ $m=0$\  and $
 \int_{-L}^L e^{\I \tilde{g}_L(x)} \sin\!\left(\frac{\pi m}{2L}x\right) \di x
     = \I \frac{4L}{\pi m} \sin\!\left(\Phi_L(L)\right)$
for odd $m$.
Merging these results yields
\begin{align*}
 \langle\varphi_{L,j},e^{\I \tilde{g}_L}\varphi_{L,k}\rangle = 
\begin{cases}
  \cos\!\left(\Phi_L(L)\right) & \text{for}\ j=k\,, \\
  \frac{2 \I}{\pi}\sin\!\left(\Phi_L(L)\right) \big[ \frac{1}{j+k} + \frac{1}{j-k} \big] & \text{for}\ j\ \text{even and}\ k\ \text{odd},\\
  \frac{2 \I}{\pi}\sin\!\left(\Phi_L(L)\right) \big[ \frac{1}{j+k} - \frac{1}{j-k} \big] & \text{for}\ j\ \text{odd and}\ k\ \text{even, and}\\
  0 & \text{for}\ j\pm k \ \text{even}.
\end{cases}
\end{align*}
In symbolic notation the resulting determinant has the form
\begin{align*}
  \left|\widetilde{D}_{N,L}\right|^2 & = \left|\det\!
\begin{pmatrix}
  c_1c_1 & c_1s_2    & c_1c_3    & \cdots & c_1s_{2M} \\
  s_2c_1 & s_2s_2    & s_2c_3    & \cdots & s_2s_{2M} \\
  \vdots &           &           &        &    \vdots \\
  s_{2M} & s_{2M}s_2 & s_{2M}c_3 & \cdots & s_{2M}s_{2M}
\end{pmatrix}
  \right|^2\\
  & = \left|\det\!
\begin{pmatrix}
  c_1c_1     & c_1c_3 & \cdots & c_1c_{2M-1}     & c_1s_2  & \cdots   & c_1s_{2M} \\
  \vdots    &        &        &               &         &         &           \\ 
  c_{2M-1}c_1 & \cdots &        & c_{2M-1}c_{2M-1} & c_{2M-1}s_2 & \cdots & c_{2M-1}s_{2M} \\
  s_2c_1     & \cdots &        & s_2c_{2M-1}     & s_2s_2    & \cdots & s_2s_{2M} \\
  \vdots    &        &        &               &           &        &           \\
  s_{2M}c_1  & \cdots &        & s_{2M}c_{2M-1}   & s_{2M}s_2 & \cdots  & s_{2M}s_{2M}
\end{pmatrix}
\right|^2
\end{align*}
where we interchanged columns and rows for the last identity.
Due to the particular structure of the matrix elements we obtain
\begin{align*}
  \left|\widetilde{D}_{N,L}\right|^2 = \left| \det\!
\begin{pmatrix}
  \cos\!\left(\Phi_L(L)\right)\id                & -\frac{2\I}{\pi}\sin\!\left(\Phi_L(L)\right) H_M^r \\
  -\frac{2\I}{\pi}\sin\!\left(\Phi_L(L)\right) H_M^l & \cos\!\left(\Phi_L(L)\right)\id
\end{pmatrix}
\right|^2\,.
\end{align*}
Here the two $M\times M$-matrices $H_M^l$ and $H_M^r$ have the entries
$\frac{1}{2}\left(\frac{1}{j+k-\frac{1}{2}} + \frac{1}{j-k+\frac{1}{2}}\right)$
and
$\frac{1}{2} \left(  \frac{1}{j+k-\frac{1}{2}} - \frac{1}{j-k-\frac{1}{2}} \right)$, respectively.
The determinant is further reduced to
\begin{align*}
  \left|\widetilde{D}_{N,L}\right|^2 &= \left|\det\!\left( \cos^2\!\left(\Phi_L(L)\right) \id + \frac{4}{\pi^2}\sin^2\!\left(\Phi_L(L)\right) H_M^lH_M^r \right) \right|^2\\
            &= \left| \det\!\left(\id - \frac{4}{\pi^2}\sin^2\!\left(\delta_L\right) K_M\right) \right|\,,
\end{align*}
using the $\pi$-periodicity of $\sin^2$ and $\Phi_L(L) = n_L \pi + \delta_L$ with $n_L \in \Z\,,\ |\delta_L| \leq \frac{\pi}{2}$.
The matrix $K_M$ is defined by
$
  (K_M)_{jk} := jk\sum\limits_{l=M+1}^\infty \frac{1}{\left[ \left(l-\frac{1}{2}\right)^2-j^2\right]\left[ \left(l-\frac{1}{2}\right)^2-k^2\right]}$.
Using the elementary formula
$
  \frac{jk}{\left[ \left(l-\frac{1}{2}\right)^2-j^2\right] \left[ \left(l-\frac{1}{2}\right)^2-k^2\right]}
    = \frac{1}{4}\left[ \frac{1}{l-\frac{1}{2}-j} - \frac{1}{l-\frac{1}{2}+j}\right]
                 \left[ \frac{1}{l-\frac{1}{2}-k} - \frac{1}{l-\frac{1}{2}+k}\right]
$
we obtain the decomposition $K_M= K_M^{--} + K_M^{+-} + K_M^{-+} + K_M^{++}$\,. Here the matrices are given by
$  \left(K_M^{--}\right)_{jk} := \frac{1}{4} \sum\limits_{l=M+1}^\infty \frac{1}{l-\frac{1}{2}-j} \frac{1}{l-\frac{1}{2}-k}$\,,
$  \left(K_M^{+-}\right)_{jk} := -\frac{1}{4} \sum\limits_{l=M+1}^\infty \frac{1}{l-\frac{1}{2}+j} \frac{1}{l-\frac{1}{2}-k}$\,,
$  \left(K_M^{-+}\right)_{jk} := -\frac{1}{4} \sum\limits_{l=M+1}^\infty \frac{1}{l-\frac{1}{2}-j} \frac{1}{l-\frac{1}{2}+k}$ and
$  \left(K_M^{++}\right)_{jk} := \frac{1}{4} \sum\limits_{l=M+1}^\infty \frac{1}{l-\frac{1}{2}+j} \frac{1}{l-\frac{1}{2}+k}$\,.
Introducing the operators
$A := \left( \frac{1}{j-\frac{1}{2}-k} \right)_{j,k\in\N}$,
$B := \left( \frac{1}{j-\frac{1}{2}+k} \right)_{j,k\in\N}$ and
the orthogonal projection $P_M := \sum\limits_{k=1}^M \langle \varphi_{L,k},\,\cdot\,\rangle \varphi_{L,k}$
we obtain
\begin{align*}
  K_M^{--} &= \tfrac{1}{4} P_MA^*(\id-P_M)AP_M\,, \quad K_M^{+-} = -\tfrac{1}{4} P_MB(\id-P_M)AP_M\,, \\
  K_M^{-+} &= -\tfrac{1}{4} P_MA^*(\id-P_M)BP_M\,, \quad K_M^{++} = \tfrac{1}{4} P_MB(\id-P_M)BP_M\,.
\end{align*}
Note that $B$ is self-adjoint, $B^*=B$.
In order to estimate the determinant we need the trace norms of these operators.
To begin with,
\begin{align*}
  \left\|K_M^{++}\right\|_1 = \frac14\left\|P_MB(\id-P_M)BP_M\right\|_1
 = \frac14\sum_{j=1}^M\sum_{k=M+1}^\infty \frac{1}{\left(j+k-\frac{1}{2}\right)^2}
 = O(1)\,.
\end{align*}
Analogously, we obtain
\begin{align*}
  \left\|K_M^{--}\right\|_1 
    & = \frac14\left\|P_MA^*(\id-P_M)AP_M\right\|_1 
       = \frac14\sum_{j=1}^M\sum_{k=M+1}^\infty \frac{1}{\left(j-\frac{1}{2}-k\right)^2}\\
    & = \frac14\sum_{j=1}^M \sum_{k=j}^\infty \frac{1}{\left(k+\frac{1}{2}\right)^2}
       = \frac14\ln(N) + O(1)
\end{align*}
for large $N=2M$. The mixed terms are estimated via the Cauchy--Schwarz inequality:
\begin{align*}
  \left\|K_M^{\pm\mp}\right\|_1 \leq \frac14\left\|P_MA^*(\id-P_M)\right\|_2 \left\|P_MB(\id-P_M)\right\|_2 = O\big(\sqrt{\ln(N)}\big)\,.
\end{align*}
Since $K_M^{--}$ gives the leading term for large $N$, we need the operator norm of $K_M^{--}$.
To this end we first rewrite the matrix elements as
\begin{align*}
  \frac{1}{4} \sum_{l=M+1}^\infty \frac{1}{l-\frac{1}{2}-j} \frac{1}{l-\frac{1}{2}-k}
     = \frac{1}{4} \sum_{l=1}^\infty \frac{1}{l+M -\frac{1}{2}-j} \frac{1}{l+M-\frac{1}{2}-k}\,.
\end{align*}
Next we define the operator $\Theta_M:\C^M\to\C^M$ by $(\Theta_M u)_j := u_{M-j}$ for any $j=1,\ldots,M$.
Obviously, $\Theta_M^2=\id$ and $\Theta_M^*=\Theta_M$.
In particular, $\Theta_M$ is unitary. 
Furthermore, we need the infinite Hilbert matrix $H_{-\frac12} := \left( \frac{1}{j+k-\frac12} \right)_{j,k\in\N}$.
Then $K_M^{--} = \frac{1}{4}\Theta_M P_M H_{-\frac12}^2 P_M \Theta_M$ and
$\left\| K_M^{--}\right\|_{\text{op}} \leq \frac14 \left\| H_{-\frac12} \right\|^2_{\text{op}} = \frac{\pi^2}{4}$
by the well-known result $\left\| H_{-\frac12}\right\|_{\text{op}} = \pi$ \cite{Sch-1911, Jam-2006}.
We conclude that the operator $\id-\frac{4}{\pi^2}\sin^2(\delta_L) K_M^{--}$ is invertible since
$\left\| \sin^2(\delta_L) K_M^{--} \right\|_{\mathrm{op}} < \frac{\pi^2}{4}$ for $0 < |\delta_L| < \frac{\pi}{2}$.
Therefore this operator yields the leading term and we have
\begin{align*}
&\left| \det\!\left( \langle \varphi_{L,j}, \psi_{L,k} \rangle \right)_{j,k=1,\ldots,N} \right|^2 = \left| \det\!\left( \id - \frac{4}{\pi^2} \sin^2\!\left(\delta_L\right) K_M^{--}\right)\right|^2\\ 
&\qquad \times \left|\det\!\left( \id - \left[\id - \frac{4}{\pi^2} \sin^2\!\left(\delta_L\right) K_M^{--}\right]^{-1} \frac{4}{\pi^2} \sin^2\!\left(\delta_L\right) \left[K_M^{++} + K_M^{+-} + K_M^{-+}\right] \right) \right|^2\,.
\end{align*}
The second determinant on the r.\,h.\,s. is bounded (with similar arguments as in the proof of the Lemma~\ref{lem-main}). 
The asserted upper bound is obtained by using the inequality $\det (\id+A) \leq  \exp\!\left(\tr A\right)$,
\begin{align*}
\left| \det\!\left( \id - \frac{4}{\pi^2} \sin^2\!\left(\delta_L\right) K_M^{--}\right)\right|^2 &\leq \exp\!\left(- \frac{8}{\pi^2}\sin^2(\delta_L) \left\|K_M^{--}\right\|_1\right)\\ 
&= \exp\!\left(-\frac{2\sin^2(\delta)}{\pi^2} \ln(N) + O\left(\sqrt{\ln(N)}\right)\right)\,,
\end{align*}
which completes the proof of Theorem~\ref{thm: upper bound}.

\section{Singling out the effective flux - proof of Lemma~\ref{lem-main}}\label{sec: decomposition}

In the following we prove the first assertion in Lemma~\ref{lem-main}.
As a first step, we show that $T_N(\e^{\I \tilde{g}_L})$ has a bounded inverse.
This is easily concluded from the following general properties of a Toeplitz matrix.
Due to the special structure these properties of a Toeplitz matrix are taken over from its symbol.
We state them without proof.

\begin{prop}\label{prop-1}
\begin{enumerate}[(i)]
\item The map $f\mapsto T_N(f)$ is linear.
\item If $f$ is real valued, $T_N(f)$ is self-adjoint, i.\,e. $T_N(f)^*=T_N(f)$.
\item If $f\geq 0$ ($f\leq 0$), then $T_N(f)$ is positive (negative) semidefinite.\label{prop-1-3}
\item If $\pm\re (f)\geq \delta>0$\,, then $T_N(f)$ is invertible with $\|T_N(f)^{-1}\|\leq \frac{1}{\delta}$.\label{prop-1-4}
\end{enumerate}
\end{prop}

\begin{rem}
Assertion~{\itshape{(\ref{prop-1-4})}} is a corollary of the Lax--Milgram Theorem.
\end{rem}

\subsection{Periodic boundary conditions}

The real part of the symbol of $T_N\big(\e^{\I \tilde{g}_L}\big)$ is
\begin{align*}
  (-1)^{n_L}\re\big(\e^{\I \tilde{g}_L (x)}\big) = (-1)^{n_L}\re\big( \e^{\pm\I \Phi_L(L) - \I \frac{\delta_L}{L}x} \big) \geq \cos(\delta_L)
\end{align*}
for all $x\in[-L,L]$, since $|\delta_L| < \frac{\pi}{2}$ and
$\tilde{g}_L(x) \in \left] - n_L\pi-\delta_L,-n_L\pi\right] \cap \left[n_L\pi,n_L\pi+\delta_L\right]$.
Then Proposition~\ref{prop-1}{\itshape{(\ref{prop-1-4})}} implies that $T_N\big(\e^{\I \tilde{g}_L}\big)$ is invertible and
\begin{align}\label{eq-boundedinverse}
  \left\| T_N\!\left(\e^{\I \tilde{g}_L}\right)^{-1} \right\| \leq \frac{1}{\left|\cos(\delta_L)\right|}\,.
\end{align}

In order to investigate $T_N\big(\e^{\I g_L}\big)$ we rewrite it in the following way:
\begin{align*}
T_N\big(\e^{\I g_L}\big) = T_N\big(\e^{\I \tilde{g}_L}\big) \Big[ \id + T_N\big(\e^{\I \tilde{g}_L}\big)^{-1}
\Big( T_N\big(\e^{\I g_L}\big) - T_N\big(\e^{\I \tilde{g}_L}\big) \Big) \Big]\,.
\end{align*}
Eventually, we want to study determinants and therefore need the trace class properties of the
difference $\Delta_N := T_N(\e^{\I g_L}) - T_N(\e^{\I \tilde{g}_L})$. Since the trace norm can be computed
most easily for semidefinite matrices, we first decompose $\Delta_N$ into hermitian matrices:
$\Delta_N = T_N(e^+) + T_N(e^-) - \I T_N(f^+) + \I T_N(f^-)$.
This decomposition is used to obtain an upper bound for the trace norm of $\Delta_N$:
\begin{align*}
\left\| \Delta_N\right\|_1 \leq \left\| T_N(e^+)\right\|_1 + \left\| T_N(e^-)\right\|_1
+ \left\| T_N(f^+)\right\|_1 + \left\| T_N(f^-)\right\|_1\,.
\end{align*}
By linearity, it suffices to decompose the respective symbols. We start by splitting the symbol into real and imaginary part
\begin{align}\label{eq-1}
  \e^{\I g_L(x)} - \e^{\I \tilde{g}_L(x)}
    =  \cos(g_L(x)) - \cos(\tilde{g}_L(x))
       + i \big[\sin(g_L(x)) - \sin(\tilde{g}_L(x) ) \big]\,.
\end{align}
We recall the trigonometric identities
$  \cos(u) - \cos(v) = -2\sin\!\left(\frac{u+v}{2}\right)\sin\!\left(\frac{u-v}{2}\right)$ and
$  \sin(u) - \sin(v) = 2\cos\!\left(\frac{u+v}{2}\right)\sin\!\left(\frac{u-v}{2}\right)$
which are used to rewrite the real and imaginary part. Then the difference $g_L(x)-\tilde{g}_L(x)$ and the
sum $g_L(x)+\tilde{g}_L(x)$ occur.
We first compute the difference and obtain
\begin{align}
  g_L(x) - \tilde{g}_L(x)
&= -\int\limits_x^{+L} a(y)\,\di y =: - \Phi_L^-(x) \ \text{ for }\ x\geq0\quad \text{ and }\label{eq-Phi-}\\
  g_L(x) - \tilde{g}_L(x)
&= \int\limits_{-L}^x a(y)\,\di y = : \Phi_L^+(x)\ \text{ for }\ x<0\,. \label{eq-Phi+}
\end{align}
For the sum we find
\begin{align*}
  g_L(x) + \tilde{g}_L(x)
&= \Phi_L^+(x)-2\frac{\delta_L}{L}x\ \text{ for }\ x\geq 0\quad \text{ and }\\
  g_L(x) + \tilde{g}_L(x)
&= -\Phi_L^-(x) - 2 \frac{\delta_L}{L}x\ \text{ for }\ x<0\,.
\end{align*}
Rewriting the trigonometric functions in the real part on the r.\,h.\,s. of Eq.~\eqref{eq-1} we
obtain
\begin{align*}
  \cos(g_L(x)) - \cos(\tilde{g}_L(x))
    &= - 2\sin\!\left( \frac{1}{2}\Phi_L^+(x) - \frac{\delta_L}{L}x\right)\sin\!\left( -\frac{1}{2}\Phi_L^-(x)\right)\ \text{ for }\ x \geq 0\ \text{ and }\\
  \cos(g_L(x)) - \cos(\tilde{g}_L(x))
    &=  -2\sin\!\left( -\frac{1}{2}\Phi_L^-(x) - \frac{\delta_L}{L}x\right)\sin\!\left(\frac{1}{2}\Phi_L^+(x)\right)\ \text{ for }\ x < 0\,.
\end{align*}
Analogous computations yield for the imaginary part
\begin{align*}
  \sin(g_L(x)) - \sin(\tilde{g}_L(x))
    &= - 2\cos\!\left(\frac{1}{2}\Phi_L^+(x) - \frac{\delta_L}{L}x\right) \sin\!\left(\frac{1}{2}\Phi_L^-(x)\right)\ \text{ for }\ x \geq0\ \text{ and }\\
  \sin(g_L(x)) - \sin(\tilde{g}_L(x))
    &= 2 \cos\!\left(\frac{1}{2}\Phi_L^-(x)+ \frac{\delta_L}{L}x\right) \sin\!\left(\frac{1}{2}\Phi_L^+(x)\right)\ \text{ for }\ x < 0\,.
\end{align*}
This yields a decomposition into four hermitian Toeplitz matrices with the symbols
\begin{align*}
  e^+ (x) &:= 2 \Theta(x) \sin\!\left(\frac{1}{2}\Phi_L^+(x)-\frac{\delta_L}{L}x\right) \sin\!\left(\frac{1}{2}\Phi_L^-(x)\right)\,,\\
  e^- (x) &:= 2 \Theta(-x) \sin\!\left(\frac{1}{2}\Phi_L^-(x)+\frac{\delta_L}{L}x\right) \sin\!\left(\frac{1}{2}\Phi_L^+(x)\right)\,,\\
  f^+(x) &:= 2\Theta(x)\cos\!\left(\frac{1}{2} \Phi_L^+(x)-\frac{\delta_L}{L}x\right) \sin\!\left(\frac{1}{2}\Phi_L^-(x)\right)\,,\\
  f^-(x) &:= 2\Theta(-x)\cos\!\left(\frac{1}{2} \Phi_L^-(x)+\frac{\delta_L}{L}x\right) \sin\!\left(\frac{1}{2} \Phi_L^+(x)\right)\,,
\end{align*}
where $\Theta$ denotes the Heaviside function, $\Theta (x) := \begin{cases} 0 & \text{ for } x<0\,,\\ 1 & \text{ for } x\geq 0\,.\end{cases}$

However, these four hermitian matrices are indefinite in general and we have to estimate their trace norm.
In view of Proposition~\ref{prop-1}{\itshape{(\ref{prop-1-3})}} we decompose the symbols in their positive and their negative part,
e.\,g.
\begin{align*}
e_{\geq}^\pm(x) := \max\!\left(0,e^\pm(x)\right) \quad \text{and} \quad e_{<}^\pm(x) := \min\!\left(0,e^\pm(x)\right)\,.
\end{align*}
Thus we split each of the four hermitian Toeplitz matrices into a positive semi-definite and a negative definite Toeplitz matrix.
We give the arguments for the decomposition of $e^+$, the estimates for the other three cases are obtained
similarly. We define the sets 
\begin{align*}
I_{\geq} := \left\{ x \in [0,L]\mid e^+(x)\geq0\right\} \quad \text{and}\quad
I_{<} := \left\{ x \in [0,L]\mid e^+(x)<0\right\}
\end{align*}
which yields a disjoint partition of the interval $[0,L]$. Note that $e^+(x)=0$ for all $x<0$.
Recall the definition of $\Phi_L^-$ in \eqref{eq-Phi-} and the simple inequalities $|\sin(x)|\leq 1$ and $|\sin(x)|\leq |x|$. 
Using them we find for the matrix elements
\begin{align*}
  \left|\left(T_N(e_{\geq}^+)\right)_{kk}\right| 
   &\leq \frac{1}{2L} \int\limits_{I_{\geq}} 2\,\left|\sin\!\left(\frac{1}{2}\Phi_L^+(x) - \frac{\delta_L}{L}x\right)\right|
                          \left| \sin\!\left(\frac{1}{2}\Phi_L^-(x)\right)\right| \di x \\
   &\leq \frac{1}{2L} \int\limits_{I_{\geq}} \int\limits_x^L \left|a(y)\right| \di y\, \di x
\end{align*}
and analogously
\begin{align*}
 \left|\left(T_N(e_{<}^+)\right)_{kk}\right| \leq \frac{1}{2L} \int\limits_{I_{<}} \int\limits_x^L \left|a(y)\right| \di y\, \di x\,.
\end{align*}
Since $I_{\geq} \dot{\cup} I_{<} = [0,L]$ and $\int_0^L \int_x^L \left|a(y)\right| \di y\, \di x = \int_0^L y \left|a(y)\right| \di y$, we have
\begin{align*}
\left\| T_N(e^+) \right\|_1 \leq \left\| T_N(e_{\geq}^+)\right\|_1 + \left\| T_N(e_{<}^+)\right\|_1
\leq \frac{N}{2L} \int\limits_0^L y \left|a(y)\right| \di y
\end{align*}
Bounds for the matrix elements of $T_N(e^-)$ and $T_N(f^\pm)$ are found likewise: 
\begin{align*}
  \left\| T_N(f^+) \right\|_1 \leq \frac{N}{2L} \int\limits_0^L y \left|a(y)\right| \di y \ \text{and}\
  \left\| T_N(e^-) \right\|_1+\left\| T_N(f^-) \right\|_1\leq \frac{N}{L} \int\limits_{-L}^0 \left|ya(y)\right| \di y\,.
\end{align*}
Altogether we have 
\begin{align}\label{eq-boundeddiff}
\left\| \Delta_N \right\|_1 \leq \frac{N}{L} \int\limits_{-L}^L \left|ya(y)\right| \di y\,.
\end{align}

Merging the estimates given in Eqs.~\eqref{eq-boundedinverse} and \eqref{eq-boundeddiff} we obtain the bound
\begin{align*}
  \left\| T_N(\e^{\I \tilde{g}_L})^{-1}\left(T_N(\e^{\I g_L})- T_N(\e^{\I \tilde{g}_L})\right)\right\|_1
     &\leq \left\| T_N(\e^{\I \tilde{g}_L})^{-1}\right\| \left\|\Delta_N\right\|_1\notag\\
     &\leq \frac{1}{\left|\cos(\delta_L)\right|} \frac{N}{L}\int\limits_{-L}^L \left|ya(y)\right| \di y\,.
\end{align*}
In order to ensure that the determinant is non-zero, the right-hand side can be made as small as we wish,
e.\,g. if the quotient $\frac{N}{2L}$ is considered to be small.
We conclude that the determinant is estimated by
\begin{align*}
  0 < \delta_1 \leq \left|\det\!\left(\id + T_N(\e^{\I \tilde{g}_L})^{-1}\left(T_N(\e^{\I g_L})-T_N(\e^{\I \tilde{g}_L})\right) \right)\right|^2 \leq \delta_2<\infty
\end{align*}
with appropriate constants $\delta_1$ and $\delta_2$ which finishes the proof of Lemma~\ref{lem-main}. 

\begin{rem}
The bounds on the matrix elements of $T_N(e^\pm)$ and $T_N(f^\pm)$ hold for a general orthonormal basis
of essentially bounded functions:
\begin{align*}
&\left|\left(T_N(e^+)\right)_{jk}\right| + \left|\left(T_N(f^+)\right)_{jk}\right|  \leq 2\|\varphi_j\|_\infty \|\varphi_k\|_\infty \int\limits_0^L y \left|a(y)\right| \di y\quad \text{and}\\
&\left|\left(T_N(e^-)\right)_{jk}\right| + \left|\left(T_N(f^-)\right)_{jk}\right| \leq 2\|\varphi_j\|_\infty \|\varphi_k\|_\infty \int\limits_{-L}^0 \left|y\right| \left|a(y)\right| \di y\,.
\end{align*}
\end{rem}

\subsection{Dirichlet boundary conditions}

The proof for Dirichlet boundary conditions is similar to the one shown above,
except that the term $\frac{\delta_L}{L}x$ is not present in the Dirichlet case.
The difference is that now the decomposition
\begin{align*}
T_N(\e^{\I g_L}) - T_N(\e^{\I \tilde{g}_L}) = E_N - \I F_N^+ + \I F_N^-
\end{align*}
is into three positive semidefinite Toeplitz matrices, instead of four, with symbols
\begin{align*}
  e (x) &:= 2 \Theta(-x) \sin\!\left(\tfrac{1}{2}\Phi_L^-(x)\right) \sin\!\left(\tfrac{1}{2}\Phi_L^+(x)\right)\,,\\
  f^+(x) &:= 2\Theta(x)\cos\!\left(\tfrac{1}{2} \Phi_L^+(x)\right) \sin\!\left(\tfrac{1}{2}\Phi_L^-(x)\right)\,,\\
  f^-(x) &:= 2\Theta(-x)\cos\!\left(\tfrac{1}{2} \Phi_L^-(x)\right) \sin\!\left(\tfrac{1}{2} \Phi_L^+(x)\right)\,.
\end{align*}
The assertion of the Lemma~\ref{lem-main} then follows by adopting the arguments of the previous proof.

\subsection*{Acknowledgment}
We would like to thank M. Gebert, H. K{\"u}ttler, H. Leschke, S. Morozov, P. M{\"u}ller and G. Raikov for fruitful discussions.
We also thank the anonymous referees for carefully reading the manuscript and constructive comments.

\bibliographystyle{KOS-myabbrvnat}

\end{document}